\documentclass[aps,pra,twocolumn,secnumarabic,amsmath,amssymb,floatfix]{revtex4-2}
\usepackage{amsthm}
\usepackage{graphicx}
\usepackage{booktabs}

\theoremstyle{definition}
\newtheorem{lemma}{Lemma}[section]
\newtheorem{corollary}[lemma]{Corollary}
\theoremstyle{remark}
\newtheorem{remark}{Remark}[section]
\theoremstyle{definition}
\newtheorem{ruleenv}{Rule}[section]
\newtheorem*{defn}{Definition}

\begin{document}

\title{Noise redistribution by stabilizer codes:\\
an organizing view of concatenated three-qubit codes}
\author{Long Huang}
\email{huangl@sicnu.edu.cn}
\affiliation{College of Physics and Electronic Engineering, Sichuan Normal University, Chengdu 610101, China}

\begin{abstract}
We formulate stabilizer-code error correction as a code- and recovery-dependent transformation of the noise channel. For a fixed recovery map, syndrome projection produces a logical channel that is block diagonal with respect to the syndrome partition, whose diagonal Pauli weights are fixed by the code-induced error classes, with no physical Pauli twirl required. Concatenation is then iterated redistribution, with the dominant logical error guiding the next code. The bit-flip and phase-flip three-qubit codes with standard recoveries supply a closed-form benchmark, and logically inequivalent recoveries appear as a design degree of freedom for shaping the channel presented to the next decoder. For amplitude damping, the rules select $C_2(\hat\sigma_y)C_1(\hat\sigma_x)$ within this family: one level converts the damping into a nearly pure logical-$X$ channel at the price $f_1=1-3\gamma/4$, the second level suppresses that channel to $f_2=1-\tfrac94\gamma^2$, and the concatenated protocol overtakes the unencoded channel below $\gamma_c=0.2905$ ($f_c=0.8485$) under ideal encoding and recovery; with a per-gate error rate $\varepsilon$ the crossover survives only for $\varepsilon$ below about $5\times10^{-4}$, and then only inside a window of $\gamma$ that contracts as $\varepsilon$ grows. Across the $n$-qubit repetition family with odd $n$, the suppression order rises as $\gamma^{w+1}$ with $w=(n-1)/2$, while the gate budget grows as $n^L$. Correlated noise requires joint class statistics, whereas coherent noise retains intra-class coherences that make the same transformation usable in reverse as a tomographic fingerprint, whose local identifiability is established in a companion paper. The construction remains an analytic benchmark, not a universal prescription for an inner code beneath a surface code.
\end{abstract}
\maketitle

\section{Introduction}
\label{sec:intro}

Noise-adapted three-qubit codes for amplitude damping have been constructed using relaxed Knill--Laflamme conditions and probabilistic recovery~\cite{ref:ad3qubit}, and related noise-strength-adapted code families have been obtained by machine-learning search~\cite{ref:nsa}; recent hardware demonstrations have reached break-even logical lifetimes against native damping noise~\cite{ref:adhw}. Those constructions tailor the code to the noise model. In contrast, our framework keeps the standard bit-flip and phase-flip codes and their standard recoveries fixed, and derives concatenation rules from the syndrome partition.

More generally, it is well established that a code does not merely suppress noise: encoding and syndrome extraction also change the structure of the effective logical channel. Beale \emph{et al.} proved that stabilizer encoding and syndrome measurement decohere errors toward probabilistic Pauli noise even without recovery~\cite{ref:beale}. Iverson and Preskill analyzed the coherence of logical channels under coherent rotations~\cite{ref:iverson}, while Behrends and B\'eri developed a detailed account of logical-noise coherence in the surface code, including the fully coherent limit~\cite{ref:surfacecoh}. Thus, block diagonalization and the suppression of particular coherences are not, by themselves, new observations. Codes have also been used deliberately as probes and as elements of noise-adaptive design. Quantum-error-correction-based characterization of dynamics (QECCD) uses code states and stabilizer-measurement statistics for process characterization~\cite{ref:qeccd,ref:qeccd2}. More recently, learning-based concatenation estimates the effective channel at each level and adapts the subsequent encoder accordingly~\cite{ref:learningconcat}, while protected stabilizer-eigenspace (PROSE) encoding selects a codespace to exploit correlated coherent noise~\cite{ref:cohframe}. These works demonstrate the practical value of treating the effective channel as a level-dependent object rather than as a fixed Pauli error rate.

The viewpoint developed here also has direct continuity with the author's earlier results. The concatenated five-qubit code was shown to construct a near-ideal effective channel from arbitrary noise, with the required concatenation size governed by the initial entanglement fidelity~\cite{ref:huang2015}. A subsequent analysis of fluctuating channels showed that the average effective channel depends only on the average physical noise; it also found convergence toward depolarizing or Pauli-dominated channels depending on the code, together with an exponential decay of fluctuations under concatenation~\cite{ref:huang2019}. Returning to the minimal three-qubit codes made the directional structure explicit: the bit-flip and phase-flip codes polarize effective channels toward distinct Pauli axes, their concatenation order matters, and a set of concatenation rules for general independent noise was summarized~\cite{ref:huang2022}.

Because the present manuscript shares its physical setting with Refs.~\cite{ref:beale,ref:huang2015,ref:huang2019,ref:huang2022}, we state the division of labour explicitly; Table~\ref{tab:compare} gives the item-by-item comparison.

\emph{What is not new.} The observation that stabilizer encoding and syndrome measurement decohere the noise toward a Pauli form is due to Ref.~\cite{ref:beale}; for a fixed recovery that is a unitary on each syndrome sector, inserting the recovery can neither create nor destroy inter-sector coherence, so the block structure of Lemma~\ref{lem:redistribution}(i) and the twirl independence (iv) are consequences of the mechanism identified there. The polarization of the two three-qubit codes toward distinct Pauli axes, the concatenation rules in the form used here, and the use of the polarized direction to select the next code are the content of Ref.~\cite{ref:huang2022}. The class-counting identity behind Lemma~\ref{lem:redistribution}(ii)--(iii) is the standard Knill--Laflamme bookkeeping.

\emph{What is new here.} (a) A statement of the redistribution lemma that keeps the recovery map explicit and separates the code-induced class partition from the recovery-dependent one (Remark~\ref{rem:recovery-design}), together with the observation that only equal-syndrome coherences survive, which yields the closed-form weight formula of Appendix~\ref{app:kraus}. (b) Closed-form recursions for the two three-qubit codes, including the $Y$-correction convention and the relabelling that it forces (Sec.~\ref{sec:ycorr} and Appendix~\ref{app:expansions})---omitting that relabelling is a silent error that shifts $\gamma_c$ from $0.2905$ to $0.2920$, i.e.\ by $0.5\%$ in $\gamma_c$ (Sec.~\ref{sec:crossover}). (c) The crossover with the unencoded channel evaluated exactly rather than to leading order, and the resulting value $\gamma_c=0.2905$. (d) A quantitative answer to the practical question left open by the idealized comparison: how accurate the gates must be for the crossover to survive (Sec.~\ref{sec:gateprec}). (e) The suppression-order versus gate-budget analysis over the $n$-qubit repetition family with an explicitly stated gate model (Sec.~\ref{sec:blocklength}). (f) A boundary analysis for coherent and correlated noise, and the identification of logically inequivalent recovery assignments as a design variable rather than a fixed convention.

The purpose of the framework is to make the channel transformation explicit, so that code selection, recovery design, concatenation, and characterization can be discussed within a common language. Accordingly, the paper does not replace the stabilizer formalism, does not identify a new class of correctable errors, and does not claim that every logical channel becomes strictly Pauli: residual intra-class coherence can remain and is itself informative.

Two scope statements are used throughout. First, all thresholds are defined in the \emph{ideal-recovery sense}: encoding, syndrome measurement, and recovery operations are assumed perfect, and $f_l$ denotes the effective channel fidelity after $l$ concatenation levels. These numbers are not fault-tolerant thresholds; Sec.~\ref{sec:gateprec} quantifies what perfect gates are worth. Second, the three-qubit codes are used because their fixed standard recoveries make the mechanism analytically transparent, not because every larger code has a unique recovery or admits the same closed-form recursion, and the three-qubit arrangement is not proposed as a universally useful inner code beneath a surface code: the fixed recovery leaves little freedom to shape the output toward the local, weakly correlated noise assumed by many surface-code decoders, and the encoding and recovery circuits add gates and correlated faults on connectivity-limited hardware. Such an architecture requires a separate hardware-specific analysis.

\begin{table*}[t]
\caption{Relation to previous work. The first four rows are the statements that this manuscript inherits; the remaining rows are the items it adds.
\label{tab:compare}}
\footnotesize
\setlength{\tabcolsep}{3pt}
\begin{tabular}{@{}p{0.27\textwidth}p{0.27\textwidth}p{0.28\textwidth}p{0.14\textwidth}@{}}
\toprule
Item & This work & Ref.~\cite{ref:beale} & Refs.~\cite{ref:huang2015,ref:huang2019,ref:huang2022}\\
\midrule
Block structure from projection & Lemma~\ref{lem:redistribution}(i), recovery explicit & identified, recovery-free & ---\\
Twirl independence of the weights & Lemma~\ref{lem:redistribution}(ii),(iv); App.~\ref{app:kraus} & follows from the same mechanism & ---\\
Class counting $\Pr(\bar P)=\sum\Pr(E)$ & Lemma~\ref{lem:redistribution}(iii) & --- & bookkeeping\\
Polarization toward $\bar Z$ and $\bar X$ & Corollary~\ref{cor:polarization} & --- & stated\\
Concatenation rules & Rules~\ref{rule:physical}--\ref{rule:logical}, family-optimal & --- & summarized\\
Five-qubit, fluctuating channels & --- & --- & Refs.~\cite{ref:huang2015},~\cite{ref:huang2019}\\
Recursions \eqref{eq:rec-C1-X}, \eqref{eq:rec-C2-Z} & given explicitly & --- & not given\\
$Y$ correction and forced relabelling & Sec.~\ref{sec:ycorr}, Eq.~\eqref{eq:rec-C2-Y} & --- & not explicit\\
$f_1=1-3\gamma/4$, $f_2=1-9\gamma^2/4$ & Eqs.~\eqref{eq:lvl1a}, \eqref{eq:f2} & --- & ---\\
Crossover with unencoded channel & $\gamma_c=0.29050$, exact & --- & qualitative only\\
Gate-accuracy tolerance $\varepsilon_{\max}$ & Sec.~\ref{sec:gateprec}, Table~\ref{tab:epsmax} & --- & ---\\
Order $w+1$ versus budget $n^L$ & Sec.~\ref{sec:blocklength}, Figs.~\ref{fig:cascade}, \ref{fig:budget} & --- & ---\\
Recovery as a design variable & Remarks~\ref{rem:recovery-design}, \ref{rem:noiseaware} & --- & ---\\
Coherent and correlated scope & Sec.~\ref{sec:extension} & Refs.~\cite{ref:iverson,ref:surfacecoh} & ---\\
\bottomrule
\end{tabular}
\end{table*}

\section{Noise redistribution by syndrome projection}
\label{sec:redistribution}

Let \(C\) be a stabilizer code with codespace projector \(P\) and syndrome projectors \(\Pi_s\). For a single-qubit channel \(\mathcal{N}\), one round of QEC on \(n\) physical qubits is described by
\begin{equation}\label{eq:eff}
\mathcal{E}_{\mathrm{eff}}(\rho_L)
=
\sum_s U_s \Pi_s \mathcal{N}^{\otimes n}\!\left(\rho_L \otimes |0\rangle\langle 0|_{\mathrm{anc}}\right) \Pi_s U_s^\dagger,
\end{equation}
where \(U_s\) is the recovery unitary associated with syndrome \(s\), and the trace over the ancilla is implicit. The recovery map assigns to every correctable physical error \(E\) a logical Pauli operator \(\bar P_E \in \{\bar I,\bar X,\bar Y,\bar Z\}\); for each logical Pauli \(\bar P\) the corresponding \emph{error class} is \(\mathcal{C}_{\bar P} = \{E : \bar P_E = \bar P\}\).

Write the physical noise in the Pauli error basis, \(\mathcal{N}^{\otimes n}(\rho)=\sum_{\mu\nu}\chi_{\mu\nu}E_\mu\rho E_\nu^\dagger\), and let \(s(\mu)\) be the syndrome of \(E_\mu\). Since \(\Pi_sE_\mu P=\delta_{s,s(\mu)}\Pi_sE_\mu P\), only pairs with equal syndrome contribute to Eq.~\eqref{eq:eff}. This is the precise sense in which the effective channel is block diagonal, and it is the only place where the syndrome partition enters.

\begin{lemma}[Syndrome-projection noise redistribution]
\label{lem:redistribution}
Let \(C\) be a stabilizer code, \(\mathcal{N}\) a single-qubit channel, and \(U_s\) a fixed recovery. After one round of QEC the effective logical channel takes the form
\begin{equation}\label{eq:chi}
\begin{aligned}
\mathcal{E}_{\mathrm{eff}}(\rho_L)&=\sum_{\bar P\bar P'}\chi^{L}_{\bar P\bar P'}\,\bar P\rho_L\bar P'^{\dagger},\\
\chi^{L}_{\bar P\bar P'}&=\!\!\sum_{\substack{\mu\in\mathcal{C}_{\bar P},\ \nu\in\mathcal{C}_{\bar P'}\\ s(\mu)=s(\nu)}}\!\!\chi_{\mu\nu}.
\end{aligned}
\end{equation}
It then has the following properties.
\begin{enumerate}
\item[(i)] \emph{Block structure.} Coherences survive only between error-basis elements of the same syndrome class; every term of Eq.~\eqref{eq:chi} carries \(\delta_{s(\mu),s(\nu)}\). In particular \(\mathcal{E}_{\mathrm{eff}}\) is block diagonal with respect to the syndrome partition of the error basis.
\item[(ii)] \emph{Diagonal weights.} The diagonal Pauli weights \(w_{\bar P}\equiv\chi^{L}_{\bar P\bar P}\) depend only on the code, on the fixed recovery, and on the physical error statistics; they are independent of any coherent part of \(\mathcal{N}\) that connects different syndrome classes.
\item[(iii)] \emph{Class counting.} For Pauli-diagonal physical noise, \(\chi_{\mu\nu}=p_\mu\delta_{\mu\nu}\) and therefore
\begin{equation}\label{eq:classcount}
w_{\bar P}=\Pr(\bar P)=\sum_{E\in\mathcal{C}_{\bar P}}\Pr(E).
\end{equation}
\item[(iv)] \emph{No physical twirl is needed.} A logical Pauli twirl removes the off-diagonal \(\chi^{L}_{\bar P\bar P'}\) with \(\bar P\neq\bar P'\) and leaves every \(w_{\bar P}\) unchanged. The residual coherences that the twirl removes are precisely the intra-sector, inter-class ones.
\end{enumerate}
\end{lemma}

\begin{defn}
The \emph{diagonal Pauli weights} of a logical channel \(\Lambda_L\) are the weights \(w_{\bar P}\) in the Pauli decomposition of its logical twirl, \(\Lambda_L^{\mathrm{tw}}(\rho)=\sum_{\bar P}w_{\bar P}\bar P\rho\bar P^\dagger\). Equivalently they are read off from the diagonal of the Pauli transfer matrix \(T_{ab}=\tfrac12\operatorname{Tr}[\bar\sigma_a\Lambda_L(\bar\sigma_b)]\) together with the trace-preserving condition \(T_{00}=1\). The identity weight is \(w_{\bar I}=\tfrac14\sum_aT_{aa}\) for any channel, unital or not; for a unital qubit channel the remaining weights follow from \(w_{\bar P}=\tfrac12\bigl(T_{\bar P\bar P}+1-2w_{\bar I}\bigr)\), whereas for a non-unital one the translation vector \(T_{a0}\) also contributes. All transfer matrices quoted below use this convention unless stated otherwise.
\end{defn}

\begin{proof}
Substituting the error-basis expansion into Eq.~\eqref{eq:eff} and using \(\Pi_sE_\mu P=\delta_{s,s(\mu)}\Pi_sE_\mu P\) gives Eq.~\eqref{eq:chi} directly; the recovery is a fixed unitary on each sector, so \(U_sE_\mu P=\bar P_\mu P\) up to a stabilizer element, which acts trivially on the codespace. Property (i) is the equal-syndrome constraint, (ii) and (iii) are immediate from Eq.~\eqref{eq:chi}, and (iv) follows because a Pauli twirl only recombines the \(\bar P\neq\bar P'\) entries while leaving each \(\chi^{L}_{\bar P\bar P}\) untouched. The full Kraus-level derivation, including the case of a non-Pauli \(\mathcal{N}\), is given in Appendix~\ref{app:kraus}.
\end{proof}

\begin{remark}[Relation to Ref.~\cite{ref:beale}]
\label{rem:beale}
Property~(i) is not a new physical mechanism: it is the statement that inserting a fixed unitary recovery into a syndrome-projected channel cannot create coherence between sectors, which is the situation analyzed in Ref.~\cite{ref:beale} for the recovery-free case. What the lemma adds is the explicit recovery-dependent class map of Eq.~\eqref{eq:chi} and the fact that the \emph{diagonal} weights are already fixed at the syndrome-projection step, before any twirl.
\end{remark}

\begin{remark}[Role of the logical twirl]
\label{rem:logical-twirl}
If a strictly Pauli logical channel is required for later analysis, the logical twirl can be applied after the recovery step; it commutes with concatenation and, by~(iv), does not change the diagonal weights. Two distinct statements must not be conflated: syndrome projection removes coherences between sectors, whereas coherences inside a sector survive an untwirled recovery and are removed only by the twirl.
\end{remark}

\begin{remark}[Code geometry versus recovery choice]
\label{rem:recovery-design}
The class partition \(\mathcal{C}_{\bar P}\) is a property of the code \emph{and} of the recovery assignment. Recovery operators that differ by stabilizer elements have the same logical action and define the same redistribution. When a decoder admits logically inequivalent assignments, however, each assignment defines a different class map, and the resulting channel can be selected to match the next decoding stage. This is the design degree of freedom exploited in Remark~\ref{rem:noiseaware}; the closed-form analysis of Secs.~\ref{sec:minimal}--\ref{sec:threshold} uses one fixed standard recovery and does not invoke it.
\end{remark}

\begin{remark}[Noise-aware decoding]
\label{rem:noiseaware}
Remark~\ref{rem:recovery-design} gives a channel-level interpretation of supplying a noise-aware decoder with a calibrated prior~\cite{ref:noiseaware,ref:decgraph}: the prior is the diagonal weight vector produced by the previous level, and recovery choice is the knob that shapes it.
\end{remark}

\begin{corollary}[Three-qubit polarization]
\label{cor:polarization}
Let \(C_1\) be the bit-flip code correcting \(X\) errors, and let \(C_2\) be the phase-flip code correcting \(Z\) errors. For an independent Pauli channel with small \(p_x,p_y,p_z\) and \(\epsilon=\max\{p_x,p_y,p_z\}\), one level of \(C_1\) accumulates the first-order logical error probability in \(\bar Z\), and one level of \(C_2\) accumulates it in \(\bar X\); the subleading logical weights are \(O(\epsilon^3)\). This follows from the closed-form recursions of Sec.~\ref{sec:minimal}, Eqs.~\eqref{eq:rec-C1-X} and \eqref{eq:rec-C2-Z}.
\end{corollary}

\section{Concatenation as iterated redistribution}
\label{sec:rules}

Lemma~\ref{lem:redistribution} shows that each level maps the noise to a channel whose diagonal weights are fixed. For the standard recoveries used here the effective channel after one level therefore has a dominant logical error direction determined by the code geometry, and concatenation admits a feedback interpretation: the output of level \(l\) is the input specification for level \(l+1\).

We denote by \(C_1\) the bit-flip code of Sec.~\ref{sec:bitflip}, whose syndrome partition polarizes the channel toward \(\bar Z\), and by \(C_2\) the phase-flip code of Sec.~\ref{sec:phaseflip}, whose syndrome partition polarizes toward \(\bar X\). In the concatenation symbol \(C_2(\hat\sigma_y)C_1(\hat\sigma_x)\) the codes are listed in order of application, the leftmost acting first: the physical noise meets the first-level code \(C_2(\hat\sigma_y)\) before the second-level code \(C_1(\hat\sigma_x)\). This is the opposite of the operator-product reading, in which the rightmost factor acts first; we state the convention explicitly because both readings are in use, and all recursions below follow the order of application.

\begin{ruleenv}[Physical level]
\label{rule:physical}
At the physical level we restrict attention to the two elementary three-qubit codes \(C_1\) and \(C_2\). This is a restriction of scope, not a claim that no other three-qubit code can be useful for amplitude damping; noise-adapted three-qubit codes that outperform the standard pair at moderate damping are known~\cite{ref:ad3qubit}.
\end{ruleenv}

\begin{ruleenv}[Initial \(Y\) correction]
\label{rule:initialY}
If the Pauli-\(Y\) probability is not negligible at the physical level, the level-one protocol used here corrects \(\hat\sigma_y\), implemented by decomposing \(Y=XZ\) and applying the corresponding \(X\)- or \(Z\)-class recovery. The code is chosen by the relative weights of \(Z\) and \(X\):
\begin{itemize}
\item if \(p_z \ge p_x\): use \(C_1\) and correct \(\hat\sigma_y\);
\item if \(p_z < p_x\): use \(C_2\) and correct \(\hat\sigma_y\).
\end{itemize}
The rationale is that, after one level, the residual channel should be polarized close to \(\bar Z\) or \(\bar X\); if \(Y\) is left uncorrected at the physical level, its components feed into both logical weights already at first order. It is important to be precise about the status of this rule. \emph{\(Y\) correction is optimal within the two-code family, but it is not necessary for second-order suppression.} A cascade that accepts the first-level conversion \(Y\to\bar Z\) and then corrects \(\bar Z\)---namely \(C_1(\hat\sigma_x)\to C_2(\hat\sigma_z)\)---also reaches \(1-f_2=O(\gamma^2)\); evaluated with the exact recursions of Sec.~\ref{sec:minimal} its leading coefficient is \(63\gamma^2/16\), against \(9\gamma^2/4\) for the protocol selected by Rule~\ref{rule:initialY}, a factor \(1.75\) worse (Sec.~\ref{sec:blocklength}). Rule~\ref{rule:initialY} is therefore a selection principle that picks the best member of the family, and statements below should be read in that sense.
\end{ruleenv}

\begin{ruleenv}[Logical level]
\label{rule:logical}
At concatenation level \(l\), let \(d_l \in \{X,Y,Z\}\) be the dominant error direction of the effective channel. The next code is selected as follows:
\begin{itemize}
\item \(d_l = Z\): apply \(C_2\) and correct \(\hat\sigma_z\);
\item \(d_l = X\): apply \(C_1\) and correct \(\hat\sigma_x\);
\item \(d_l = Y\): apply \(C_1\) if \(p_z^{(l)} \ge p_x^{(l)}\), else \(C_2\), and correct \(\hat\sigma_y\) in either case.
\end{itemize}
Equivalently, the concatenation protocol is a deterministic feedback process: the dominant logical error at level \(l\) becomes the error corrected at level \(l+1\). As with Rule~\ref{rule:initialY}, the rule identifies the natural choice within the family rather than proving uniqueness.
\end{ruleenv}

Applying the rules to amplitude damping: the Pauli-diagonal weights of that channel satisfy
\begin{equation}\label{eq:AD-pauli}
\begin{aligned}
f_0&=\frac{(1+\sqrt{1-\gamma})^2}{4},\qquad p_x=p_y=\frac{\gamma}{4},\\
p_z&=\frac{(1-\sqrt{1-\gamma})^2}{4}=\frac{\gamma^2}{16}+O(\gamma^3),
\end{aligned}
\end{equation}
so that \(p_x=p_y>p_z\) for all \(\gamma>0\). Rule~\ref{rule:initialY} therefore selects \(Y\) correction at the first level, and since \(p_z<p_x\) the code is \(C_2(\hat\sigma_y)\). After one level the channel is polarized toward \(\bar X\) (Corollary~\ref{cor:polarization}), so Rule~\ref{rule:logical} selects \(C_1(\hat\sigma_x)\) at the second level. The protocol is
\begin{equation}\label{eq:protocol}
C_2(\hat\sigma_y)\, C_1(\hat\sigma_x).
\end{equation}
The selection does not assume that the direction alternates. Sec.~\ref{sec:levels} shows that after the second level \(\bar X\) is still dominant, so the third level again uses \(C_1(\hat\sigma_x)\), and the balance passes to \(\bar Z\) only after that. For correlated or coherent noise the selection rules remain valid, because they depend only on the code geometry, but the weights must be recomputed from the joint error distribution (Sec.~\ref{sec:extension}).

\section{Minimal realization: three-qubit codes and recursions}
\label{sec:minimal}

The two three-qubit codes are the smallest non-trivial stabilizer codes and therefore provide the simplest explicit realization of the noise-redistribution lemma. Their syndrome partitions have only four classes, and the resulting logical channel can be written in closed form for arbitrary Pauli noise.

\subsection{Bit-flip code \(C_1\)}
\label{sec:bitflip}

The bit-flip code has stabilizer \(\langle Z_1Z_2,Z_2Z_3\rangle\), logical basis \(|0_{\mathcal L}\rangle=|000\rangle\), \(|1_{\mathcal L}\rangle=|111\rangle\), and logical operators \(\bar X=X_1X_2X_3\) and \(\bar Z=Z_1\equiv Z_2\equiv Z_3\), the three single-qubit \(Z\)'s being equal modulo the stabilizer. Its syndrome partition contains four classes: the no-error class, the three single-qubit \(X\) errors \(X_1,X_2,X_3\), and the degenerate partners that share the same syndromes. For a general Pauli channel
\begin{equation}\label{eq:pauli}
\varepsilon_0 = f_0 \hat{\sigma}_I + p_x \hat{\sigma}_x + p_y \hat{\sigma}_y + p_z \hat{\sigma}_z,
\end{equation}
with \(f_0+p_x+p_y+p_z=1\), the recursion for \(C_1\) correcting \(X\) errors is
\begin{equation}\label{eq:rec-C1-X}
\begin{aligned}
f_1 &= f_0^3 + 3f_0^2 p_x + 3f_0 p_z^2 + 6f_0 p_y p_z + 3p_x p_z^2, \\
p_{1z} &= p_z^3 + 3p_z^2 p_y + 3p_z f_0^2 + 6f_0 p_x p_z + 3p_y f_0^2, \\
p_{1x} &= p_x^3 + 3p_x^2 f_0 + 3p_x p_y^2 + 6p_x p_y p_z + 3f_0 p_y^2, \\
p_{1y} &= p_y^3 + 3p_y^2 p_z + 3p_y p_x^2 + 6f_0 p_y p_x + 3p_z p_x^2.
\end{aligned}
\end{equation}
For weak noise, \(f_0\simeq1\). The first-order terms in Eq.~\eqref{eq:rec-C1-X} are contained exclusively in \(p_{1z}\), while \(p_{1x}\) and \(p_{1y}\) are of order \(\mathcal O(\epsilon^3)\). Hence the effective channel is polarized toward \(\bar Z\), as stated in Corollary~\ref{cor:polarization}.

\subsection{Phase-flip code \(C_2\)}
\label{sec:phaseflip}

The phase-flip code has stabilizer \(\langle X_1X_2,X_2X_3\rangle\) and code words
\begin{equation}\label{eq:phaseflip}
\begin{aligned}
|0_{\mathcal L}\rangle &= \frac{1}{2}\bigl(|000\rangle+|011\rangle+|101\rangle+|110\rangle\bigr), \\
|1_{\mathcal L}\rangle &= \frac{1}{2}\bigl(|111\rangle+|100\rangle+|010\rangle+|001\rangle\bigr),
\end{aligned}
\end{equation}
which are the Hadamard-conjugated Greenberger--Horne--Zeilinger states, \(|0_{\mathcal L}\rangle=H^{\otimes3}(|000\rangle+|111\rangle)/\sqrt2\). Its logical operators are \(\bar X=X_1X_2X_3\equiv X_1\equiv X_2\equiv X_3\) and \(\bar Z=Z_1Z_2Z_3\), the equalities holding modulo the stabilizer. Note the asymmetry with the bit-flip code: here the single-qubit \(X\)'s are logical representatives, whereas no single-qubit \(Z\) is, because \(Z_1\) anticommutes with \(X_1X_2\) (only the product \(Z_1Z_2Z_3\) commutes with all stabilizers). Recovering by correcting the \(Z\) class polarizes the effective channel toward \(\bar X\), and the recursion is
\begin{equation}\label{eq:rec-C2-Z}
\begin{aligned}
f_1 &= f_0^3 + 3f_0^2 p_z + 3f_0 p_x^2 + 6f_0 p_y p_x + 3p_z p_x^2, \\
p_{1z} &= p_z^3 + 3p_z^2 f_0 + 3p_z p_y^2 + 6p_z p_x p_y + 3f_0 p_y^2, \\
p_{1x} &= p_x^3 + 3p_x^2 p_y + 3p_x f_0^2 + 6p_x p_z f_0 + 3p_y f_0^2, \\
p_{1y} &= p_y^3 + 3p_y^2 p_x + 3p_y p_z^2 + 6p_y f_0 p_z + 3p_x p_z^2.
\end{aligned}
\end{equation}
For weak noise the first-order terms appear only in \(p_{1x}\), so the channel is polarized toward \(\bar X\).

\subsection{Correction of \(Y\) errors, and the relabelling it forces}
\label{sec:ycorr}

When \(Y\) errors are corrected instead of \(X\) or \(Z\), the corrected Pauli axis changes, so the recursions follow from the interchanges
\begin{equation}\label{eq:interchange}
C_{1}:\ p_{x}\leftrightarrow p_{y},\qquad C_{2}:\ p_{z}\leftrightarrow p_{y},
\end{equation}
applied to the right-hand sides of Eqs.~\eqref{eq:rec-C1-X} and \eqref{eq:rec-C2-Z} \emph{together with} the corresponding relabelling of the output weights: correcting \(\sigma_y\) means that \(Y\) takes the role previously played by \(Z\) (for \(C_2\)) or by \(X\) (for \(C_1\)), so \(p_{1y}\) and \(p_{1z}\) exchange places as well. This point deserves emphasis because the two operations are easy to separate by accident. Applying the substitution to the right-hand sides \emph{only}, while keeping the output labels fixed, describes a \emph{different recovery}, a fact that also gives a useful diagnostic: as shown in Sec.~\ref{sec:crossover}, that variant shifts the crossover of the amplitude-damping cascade from \(\gamma_c=0.2905\) to \(0.2920\).

For the reader's convenience we write the \(C_2\) recursion for \(Y\) correction out in full. Substituting \(p_z\to p_y\), \(p_y\to p_z\) in Eq.~\eqref{eq:rec-C2-Z} and swapping the output labels \(p_{1y}\leftrightarrow p_{1z}\) gives
\begin{equation}\label{eq:rec-C2-Y}
\begin{aligned}
f_1 &= f_0^3 + 3f_0^2 p_y + 3f_0 p_x^2 + 6f_0 p_z p_x + 3p_y p_x^2, \\
p_{1x} &= p_x^3 + 3p_x^2 p_z + 3p_x f_0^2 + 6p_x p_y f_0 + 3p_z f_0^2, \\
p_{1y} &= p_y^3 + 3p_y^2 f_0 + 3p_y p_z^2 + 6p_y p_x p_z + 3f_0 p_z^2, \\
p_{1z} &= p_z^3 + 3p_z^2 p_x + 3p_z p_y^2 + 6p_z f_0 p_y + 3p_x p_y^2.
\end{aligned}
\end{equation}

The recovery operations for both codes can be implemented without measurement, using coherent Toffoli gates; discarding the syndrome ancilla is then equivalent at the channel level to measurement and feedback, consistent with Remark~\ref{rem:logical-twirl}.

\section{Application to amplitude damping}
\label{sec:threshold}

We work with the diagonal weights of the logical channel, that is, with its logical Pauli twirl. For any qubit channel the weight of the identity in the twirl is \(w_{\bar I}=\tfrac14\sum_aT_{aa}\) (with \(T_{00}=1\)), which is also its entanglement fidelity; since the twirl leaves the diagonal of the transfer matrix invariant [Lemma~\ref{lem:redistribution}(iv)], the twirled and untwirled channels have exactly the same \(f_l\), whether or not they are unital. They differ in the rest of the channel: the twirl sets the translation vector \(T_{a0}\) and any intra-class coherence to zero, and neither enters \(f_l\). It is that residual structure, not the fidelity, that Sec.~\ref{sec:inverse} treats.

We checked this directly by propagating the full amplitude-damping Kraus operators through the ideal recovery protocol of Sec.~\ref{sec:minimal} without any twirl, for the two-level cascade (three copies of the first-level logical channel inside the second-level code). With the convention \(T_{ab}=\tfrac12\operatorname{Tr}[\hat\sigma_a\Lambda(\hat\sigma_b)]\), the untwirled logical channel is Pauli diagonal to numerical precision---\(|T_{ab}|<10^{-15}\) for \(a\neq b\), \(a,b\neq0\)---but it is \emph{not} unital: its translation vector does not vanish, \(|T_{a0}|=2.1\times10^{-4}\) at \(\gamma=0.1\) and \(1.4\times10^{-2}\) at \(\gamma=0.3\), growing with \(\gamma\). The largest off-diagonal element of this transfer matrix is precisely this translation component, which reaches the percent level at \(\gamma=0.3\); it is a bias, not an intra-class coherence. Amplitude damping generates no coherence that survives the projection, but it does bias the logical channel, because the damping has a preferred state. The twirl removes the bias and leaves \(f_l\) untouched, so the equality of fidelities is exact. The recursion of Sec.~\ref{sec:minimal} twirls after each level and therefore drops the translation; as a result \(f_2\) differs from the exact entanglement fidelity of the untwirled cascade only at second order in it, by less than \(2\times10^{-7}\) for \(\gamma\le0.1\) in our evaluation (script provided as supplementary material).

The amplitude-damping channel has Kraus operators
\begin{equation}\label{eq:kraus}
A_0 = \begin{pmatrix} 1 & 0 \\ 0 & \sqrt{1-\gamma} \end{pmatrix}, \qquad
A_1 = \sqrt{\gamma}\, \lvert 0 \rangle \langle 1 \rvert,
\end{equation}
with Pauli-diagonal part given by Eq.~\eqref{eq:AD-pauli}. The non-diagonal coherent part does not enter the logical Pauli weights: by Lemma~\ref{lem:redistribution}(iv) cross-syndrome coherences vanish after projection, and an optional logical twirl removes the remaining intra-class coherences. The correctable statistics are therefore fully described by Eq.~\eqref{eq:AD-pauli}.

\subsection{First and second level}
\label{sec:levels}

Substituting Eq.~\eqref{eq:AD-pauli} into the \(C_2\)-correct-\(Y\) recursion, Eq.~\eqref{eq:rec-C2-Y}, gives the exact expansions
\begin{equation}\label{eq:lvl1a}
\begin{aligned}
f_1&=1-\tfrac34\gamma+\tfrac{7}{64}\gamma^3+O(\gamma^4),\\
p_{1x}&=\tfrac34\gamma-\tfrac{3}{16}\gamma^2+O(\gamma^3),
\end{aligned}
\end{equation}
together with
\begin{equation}\label{eq:lvl1b}
p_{1y}=\tfrac{3}{16}\gamma^2-\tfrac{5}{64}\gamma^3+O(\gamma^4),\qquad
p_{1z}=\tfrac{9}{64}\gamma^3+O(\gamma^4).
\end{equation}
Equations~\eqref{eq:lvl1a}--\eqref{eq:lvl1b} are exact rational expansions in \(\gamma\); the \(\gamma^2\) term of \(f_1\) vanishes identically. Appendix~\ref{app:expansions} lists the full set. The factor \(3\) in Eq.~\eqref{eq:lvl1a} is structural and does not depend on the recovery convention: in the phase-flip code every single-qubit \(X\) error has trivial syndrome, and \(X_j\) acts as \(\bar X\) on the codespace for each \(j=1,2,3\); correcting \(\sigma_y\) leaves those three errors untouched, so they contribute \(3p_x=3\gamma/4\) at first order. The \(Z\) errors, which the \(Y\) correction converts into \(X\) errors, enter only at higher order because \(p_z=O(\gamma^2)\). The one-level channel is therefore strongly polarized toward \(\bar X\), in agreement with Corollary~\ref{cor:polarization}, but its fidelity
\begin{equation}\label{eq:f1lead}
f_1=1-\tfrac34\gamma+O(\gamma^3)
\end{equation}
lies \emph{below} the unencoded value \(f_0=1-\gamma/2+O(\gamma^2)\). A single level is thus a converter rather than a suppressor, and the monotone criterion \(f_l>f_{l-1}\) cannot be met at \(l=1\).

The second level uses \(C_1\) and corrects \(\sigma_x\). Denoting the input parameters by \((f_1,p_{1x},p_{1y},p_{1z})\)---the output of level one in its own logical frame, with no further relabelling---and substituting them into Eq.~\eqref{eq:rec-C1-X} gives
\begin{equation}\label{eq:lvl2w}
\begin{aligned}
p_{2x}&=\tfrac{27}{16}\gamma^2-\tfrac{27}{16}\gamma^3+O(\gamma^4),\\
p_{2z}&=\tfrac{9}{16}\gamma^2-\tfrac{21}{32}\gamma^3+O(\gamma^4),\\
p_{2y}&=\tfrac{27}{32}\gamma^3+O(\gamma^4),
\end{aligned}
\end{equation}
so that
\begin{equation}\label{eq:f2}
f_2 = 1-\tfrac94\gamma^2+\tfrac32\gamma^3+\tfrac{45}{128}\gamma^4+O(\gamma^5).
\end{equation}
The logical error has dropped from first to second order, because the \(X\) errors inherited from the first level are exactly what the second code corrects. The dominant direction does \emph{not} alternate: the leading contribution to \(p_{2x}\), namely \(3p_{1x}^2\), arises from pairs of level-one \(X\) errors, so \(\bar X\) remains dominant by a factor three, \(p_{2x}/p_{2z}\to3\). Successive levels follow the deterministic rule without a two-cycle: the third level is again \(C_1(\hat\sigma_x)\), and only after it does the balance move to \(\bar Z\). Iterating the closed-form recursions gives \(f_3=1-\tfrac{27}{16}\gamma^2+O(\gamma^3)\): the third level buys a factor \(4/3\) in infidelity, because the \(\bar Z\) weight inherited from level two is not corrected by \(C_1\) and is merely tripled.

\subsection{Crossover with the unencoded channel}
\label{sec:crossover}

Because \(f_1<f_0\) for every \(\gamma\), a monotone criterion cannot be applied at \(l=1\): one level never pays on its own, and its role is to prepare the channel for the next one. The relevant comparison is between the concatenated protocol and no encoding at all,
\begin{equation}\label{eq:crit}
f_2 > f_0,
\end{equation}
evaluated with the full recursions rather than with their leading terms. The unencoded fidelity is exactly
\begin{equation}\label{eq:f0exact}
\begin{aligned}
f_0 &= \left(\frac{1+\sqrt{1-\gamma}}{2}\right)^{2}
= \frac{(1+\sqrt{1-\gamma})^2}{4}\\
&= 1-\frac{\gamma}{2}-\frac{\gamma^{2}}{16}+\mathcal O(\gamma^{3}),
\end{aligned}
\end{equation}
identical to the \(f_0\) of Eq.~\eqref{eq:AD-pauli} and carrying a linear term; a form without one, such as \(f=1-c\gamma^{2}\), can reproduce \(f_0\) near a single point but nowhere else. Numerically, \(f_2=f_0\) at
\begin{equation}\label{eq:gammac}
\gamma_c = 0.29050, \qquad f_c = 0.84853 .
\end{equation}
For \(\gamma<\gamma_c\) the two-level protocol outperforms the physical channel; beyond it, the encoding does not repay the noise it removes. This crossover is the central quantitative result of this section, and it is a two-level statement throughout.

Two remarks on Eq.~\eqref{eq:gammac} are in order, because the number is easy to get slightly wrong. First, the leading-order truncation is \emph{not} adequate: retaining only \(O(\gamma^2)\) in Eqs.~\eqref{eq:f2} and \eqref{eq:f0exact} would give \(\gamma_c=8/35\simeq0.229\), a \(21\%\) error, so the \(O(\gamma^3)\) terms must be kept. Second, the handover between levels one and two must respect the relabelling of Eq.~\eqref{eq:interchange}. If the level-one output is fed in with \(p_{1y}\) and \(p_{1z}\) interchanged---which is what one obtains by applying the substitution to the right-hand side of Eq.~\eqref{eq:rec-C2-Z} only, i.e.\ by the ``different recovery'' mentioned in Sec.~\ref{sec:ycorr}---the same exact evaluation returns \(\gamma_c=0.2920\) instead. The difference in \(\gamma_c\) is \(+0.0015\), or \(0.5\%\) in relative terms---against \(0.1\%\) in \(f_c\)---which is two orders of magnitude above the \(10^{-5}\) resolution at which \(\gamma_c\) is quoted, so it is systematic rather than a rounding artefact. That is the reason we have written Eq.~\eqref{eq:rec-C2-Y} out in full.

With noisy gates the comparison is between \(F_l=R_lf_l\), where \(R_l\) is the overall gate accuracy of the level-\(l\) protocol and \(R_l=1\) is the ideal-gate limit. Sec.~\ref{sec:gateprec} quantifies how much gate accuracy the crossover can absorb; fault-tolerant thresholds with noisy gates are not addressed here.

\subsection{How perfect must the gates be?}
\label{sec:gateprec}

The ideal-recovery crossover of Eq.~\eqref{eq:gammac} is a statement about the noise channel alone. Since the headline numbers are often read as thresholds, it is worth asking how much gate error the comparison tolerates. With a per-gate error rate \(\varepsilon\) and \(G\) gates, \(R=r^{G}\) with \(r=1-\varepsilon\), so the crossover condition \(R f_2>f_0\) becomes \(r>(f_0/f_2)^{1/G}\), i.e.
\begin{equation}\label{eq:epsmax}
\varepsilon<\varepsilon_{\max}(\gamma)\equiv1-\left(\frac{f_0(\gamma)}{f_2(\gamma)}\right)^{1/G}.
\end{equation}
Table~\ref{tab:epsmax} evaluates Eq.~\eqref{eq:epsmax} with the gate model of Sec.~\ref{sec:blocklength} (\(G=64\) for the three-qubit two-level cascade, \(G=174\) for the five-qubit one).

Three features are worth recording. (i) The tolerance is \emph{not} monotonic in \(\gamma\): it peaks near \(\gamma\simeq0.15\), where the fidelity gain \(f_2-f_0\) is largest relative to \(f_0\), and collapses as \(\gamma\to\gamma_c\) from below, where the gain vanishes. (ii) The required per-gate accuracy, a few times \(10^{-4}\), is one to two orders of magnitude beyond present two-qubit gate fidelities; the ideal-recovery crossover is therefore an analytic benchmark for the mechanism, not a device target, in line with the scope statement of Sec.~\ref{sec:intro}. (iii) The longer block is the more demanding one, because it spends \(2.7\) times as many gates for a \(1.5\)-fold smaller infidelity (Sec.~\ref{sec:blocklength}), so its \(\varepsilon_{\max}\) is smaller by a factor of between two and five even though its ideal-recovery fidelity is higher.

Equation~\eqref{eq:epsmax} is more informative read backwards. Because \(\varepsilon_{\max}(\gamma)\) is not monotonic in \(\gamma\), its inverse is double valued: the condition \(\varepsilon<\varepsilon_{\max}(\gamma)\) does not select an interval adjacent to \(\gamma=0\), but a \emph{window} \((\gamma_{\mathrm{lo}}(\varepsilon),\gamma_{\mathrm{hi}}(\varepsilon))\) that contracts as the gates degrade. At \(\varepsilon=0\) the window is \((0,\gamma_c)\); at \(\varepsilon=5\times10^{-4}\) it has shrunk to \((0.111,0.168)\); and it closes at \(\varepsilon=5.2\times10^{-4}\) near \(\gamma=0.139\), where the two edges meet. A single \(\gamma_c(\varepsilon)\) is therefore not defined, and \(5\times10^{-4}\) is the largest per-gate error at which the window exists at all, not the point at which the crossover drifts to zero. Read as a shift rather than a window, the upper edge is linear near \(\gamma_c\): \(\gamma_{\mathrm{hi}}=\gamma_c-k\varepsilon\) with \(k=1.5\times10^{2}\), accurate to \(5\%\) up to \(\varepsilon=10^{-4}\) but underestimating the shift by \(38\%\) at \(5\times10^{-4}\) (equivalently, overestimating \(\gamma_{\mathrm{hi}}\) by \(28\%\)), where the linear regime has already ended. The sign is negative---gate error moves the crossover to \emph{lower} \(\gamma\), not higher---while the lower edge moves the other way, upward from zero. The lower edge carries as much information as the upper one: for \(\gamma<0.111\) the cascade loses to doing nothing, because the accumulated cost of the \(G=64\) gates, \(1-r^{G}=3.15\times10^{-2}\) at \(\varepsilon=5\times10^{-4}\), exceeds the relative fidelity gain \((f_2-f_0)/f_2\), which reaches the same \(3.15\times10^{-2}\) only at \(\gamma=0.111\) on its way up and at \(\gamma=0.168\) on its way down. Note finally that the unencoded arm carries no gates, so Eq.~\eqref{eq:epsmax} is the most favourable reading for the cascade.

\begin{table}[htbp]
\caption{Largest per-gate error rate for which the concatenated protocol still beats the unencoded channel at damping \(\gamma\), encoding and recovery otherwise ideal. Columns 4 and 5 use the gate model of Eqs.~\eqref{eq:G}--\eqref{eq:grec}.
\label{tab:epsmax}}
\begin{ruledtabular}
\begin{tabular}{lcccc}
$\gamma$ & $f_0$ & $f_2$ $(n=3)$ & $\varepsilon_{\max}$, $n=3$ & $\varepsilon_{\max}$, $n=5$\\
\hline
$0.05$ & $0.97484$ & $0.99456$ & $3.1\times10^{-4}$ & $1.3\times10^{-4}$\\
$0.10$ & $0.94934$ & $0.97903$ & $4.8\times10^{-4}$ & $2.2\times10^{-4}$\\
$0.15$ & $0.92348$ & $0.95458$ & $5.2\times10^{-4}$ & $2.3\times10^{-4}$\\
$0.20$ & $0.89721$ & $0.92243$ & $4.3\times10^{-4}$ & $1.7\times10^{-4}$\\
$0.25$ & $0.87051$ & $0.88378$ & $2.4\times10^{-4}$ & $4.7\times10^{-5}$\\
\end{tabular}
\end{ruledtabular}
\end{table}

\subsection{Block length: suppression order versus gate budget}
\label{sec:blocklength}

The three-qubit code is the shortest member of a one-parameter family. The
\(n\)-qubit repetition code has \(Z\)-type stabilizers \(Z_iZ_{i+1}\) (bit-flip)
or \(X\)-type stabilizers \(X_iX_{i+1}\) (phase-flip)---the same generator
pattern for every odd \(n\). Because only \(n\) distinguishes the members, this
family isolates block length from code structure. We evaluated the complete
cascade exactly for \(n=3,5,7\) under the same ideal-recovery model.

\emph{Suppression order.} A recovery that returns \emph{every} syndrome sector to
the code space must handle the lightest error pattern in each sector. The
syndrome of the repetition code is the set of adjacent parities, so a pattern
and its complement fall in the same sector; the heaviest such minimum is the one
of weight \(\lfloor n/2\rfloor\). For odd \(n\) the required recovery weight is
therefore
\begin{equation}\label{eq:wn}
w(n)=\frac{n-1}{2},
\end{equation}
which is the same statement as \(\sum_{k\le w}\binom{n}{k}=2^{n-1}\), by the
symmetry \(\binom{n}{k}=\binom{n}{n-k}\); the equality is saturated, so
\(w=1,2,3\) for \(n=3,5,7\) and no lighter recovery is trace preserving. The
first level still converts the damping into a logical-\(X\) channel of weight
\(n\gamma/4\); the second level removes every pattern of weight \(\le w\) and
leaves only weight \(w+1\),
\begin{equation}\label{eq:order}
1-f_2^{(n)} = \Theta\!\left(\gamma^{\,w(n)+1}\right)
= \Theta(\gamma^{2}),\ \Theta(\gamma^{3}),\ \Theta(\gamma^{4}),
\end{equation}
for $n=3,5,7$ respectively.
A least-squares fit on thirty logarithmically spaced points in
\(\gamma\in[3\times10^{-3},3\times10^{-2}]\) gives log-log slopes \(1.993\),
\(2.957\) and \(3.901\) [Fig.~\ref{fig:cascade}(b)]; the deviation from the
asymptotic integers is a finite-\(\gamma\) effect and grows as the fit window is
moved to larger \(\gamma\). Block length therefore buys a higher suppression
order---but only through the recovery weight it forces, not through any gain in
distance. The code's distance \emph{against the error type it protects} is \(n\),
not \(w+1\): a weight-\(w\) recovery is what Eq.~\eqref{eq:wn} demands, and it is
the recovery, not the distance, that sets the order. Against the conjugate error
type the distance remains one.

\begin{figure*}[t]
\includegraphics[width=\textwidth]{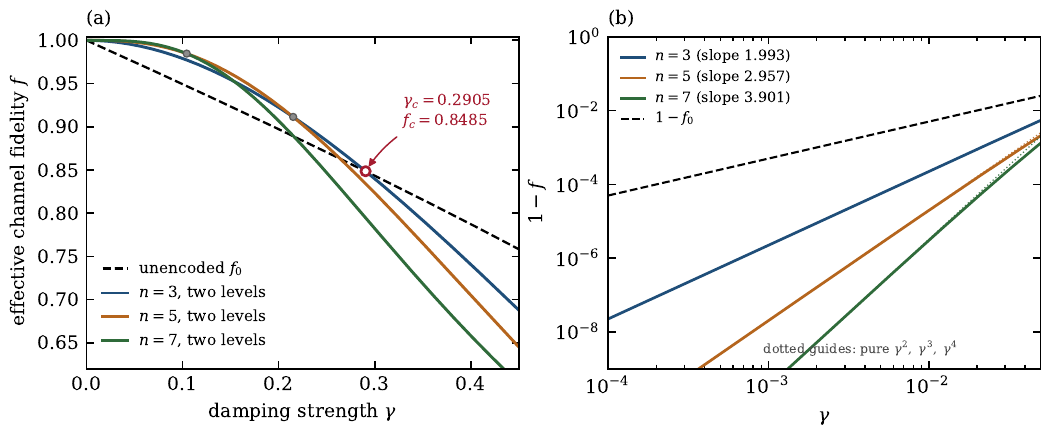}
\caption{Two-level cascade for repetition blocks of length \(n\). (a) Effective-channel fidelity against the unencoded channel (dashed). The three-qubit cascade crosses the unencoded channel at \(\gamma_c=0.2905\); the grey dots mark the block-length crossovers at \(\gamma=0.104\) (\(n=7\) overtaking \(n=5\)) and \(\gamma=0.215\) (\(n=5\) overtaking \(n=3\)). (b) The same data on log-log axes, showing the suppression orders; the dotted guides are pure \(\gamma^{2}\), \(\gamma^{3}\) and \(\gamma^{4}\).
\label{fig:cascade}}
\end{figure*}

\emph{Gate budget.} We now make the gate model explicit, since the comparison between block lengths is only meaningful once the recovery overhead is specified. The cascade uses \(n^{L}\) physical qubits and instantiates the level-\(k\) block \(\prod_{j>k}n_j\) times with \(n_j=n\) for all \(j\), so
\begin{equation}\label{eq:G}
G=\sum_{k=1}^{L}\Big(\prod_{j>k}n_j\Big)
\big(g^{\rm enc}_k+g^{\rm dec}_k+g^{\rm rec}_k\big),
\end{equation}
with \(g^{\rm enc}=n+1\) and \(g^{\rm dec}=3n-2\), which reproduce the \(4+7=11\) gates of the three-qubit block at \(n=3\). For the recovery we count one CNOT per neighbour in each of the \(n-1\) adjacent parity checks, plus one single-qubit Pauli per position of the weight-\(w\) correction pattern:
\begin{equation}\label{eq:grec}
g^{\rm rec}(n)=2(n-1)+w(n).
\end{equation}
Equation~\eqref{eq:grec} counts primitives only---single-qubit gate \(=1\), CNOT \(=1\)---and omits ancilla preparation and readout, so Eq.~\eqref{eq:G} is a lower bound on the circuit cost; a coherent-Toffoli implementation of the recovery costs more, and because its cost grows with \(w\) it favours short blocks more strongly than Eq.~\eqref{eq:grec} does. With Eqs.~\eqref{eq:G}--\eqref{eq:grec}: \(G=64\) for \((n,L)=(3,2)\), \(174\) for \((5,2)\), \(336\) for \((7,2)\) and \(208\) for \((3,3)\).

At \(\gamma=0.1\) the two-level protocols reach infidelities \(2.10\times10^{-2}\), \(1.38\times10^{-2}\) and \(1.35\times10^{-2}\) for \(n=3,5,7\): the longest block buys a factor \(1.56\) in infidelity for a factor \(5.25\) in gates. Per gate the ordering never reverses---the three-qubit cascade gains \(0.464\) in \(f\) per thousand gates, against \(0.212\) and \(0.111\)---and this is the robust part of the comparison, since it follows from the scaling described here rather than from the details of Eq.~\eqref{eq:grec}. At a fixed \emph{total} budget the ranking is more model-dependent: with Eq.~\eqref{eq:grec} the five-qubit two-level cascade reaches \(f=0.9862\) with \(174\) gates, ahead of the three-qubit three-level cascade at \(f=0.9821\) with \(208\) gates, whereas a recovery cost growing faster with \(w\) reverses this. We therefore quote the per-gate comparison as the model-independent statement and Fig.~\ref{fig:budget} as the comparison for the stated model.

\begin{figure}[t]
\includegraphics[width=\columnwidth]{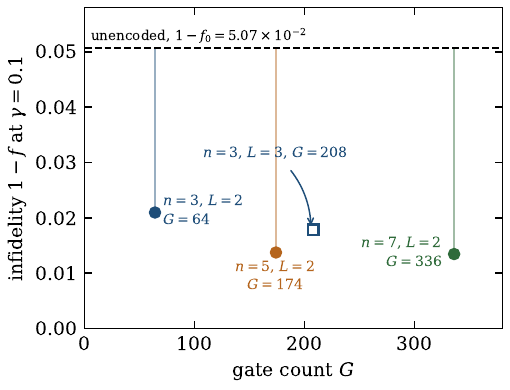}
\caption{Infidelity versus gate count at \(\gamma=0.1\) for the two-level cascades with \(n=3,5,7\) (filled circles) and the three-level three-qubit cascade (open square), using the gate model of Eqs.~\eqref{eq:G}--\eqref{eq:grec}. The dashed line is the unencoded channel.
\label{fig:budget}}
\end{figure}

Two trends cross as the noise grows. The seven-qubit cascade is the most accurate below \(\gamma\simeq0.104\), the five-qubit one between \(0.104\) and \(0.215\), and the three-qubit one above \(0.215\); the corresponding crossovers with the unencoded channel are \(\gamma_c=0.2158\), \(0.2642\) and \(0.2905\) for \(n=7,5,3\). Block length thus enters through two competing powers---order \(w(n)+1\) and gate count \(n^{L}\)---and the best choice for a given budget is the shortest block whose order still beats the unencoded channel at the operating \(\gamma\). Within the repetition family, and over almost the whole range where encoding pays at all, that is the three-qubit code. A distance-three code such as the five-qubit perfect code of Ref.~\cite{ref:fivequbit} protects arbitrary single-qubit errors at one level; that is a different axis (distance rather than block length) and is not compared here.

All numbers quoted in Sec.~\ref{sec:threshold} are produced by an evaluation script that propagates the exact recursions of Sec.~\ref{sec:minimal} in rational arithmetic (Appendix~\ref{app:expansions}) and enumerates the \(4^{n}\) Pauli strings as an independent check; it is provided as supplementary material together with the data of Figs.~\ref{fig:cascade} and \ref{fig:budget}.

\section{Boundaries of the mechanism: coherent and correlated noise}
\label{sec:extension}

Lemma~\ref{lem:redistribution} is a statement about the geometry of stabilizer codes, not about the independence or commutativity of the physical noise. Its conclusions therefore remain valid when the single-qubit channel \(\mathcal{N}\) is coherent, and the projection argument extends to noise correlated across qubits.

\subsection{Coherent noise}

For a coherent single-qubit channel the effective logical channel may contain non-Pauli coherences. Lemma~\ref{lem:redistribution} guarantees three facts: (1) the diagonal weights of \(\mathcal{E}_{\mathrm{eff}}\) depend only on the code structure and on the Pauli-diagonal part of the physical channel; (2) coherences between different syndrome classes are removed by the projection, regardless of the physical coherence; (3) residual coherences survive only within each syndrome class and can be removed by an optional logical twirl.

For the purpose of selecting the next code in a concatenation, only the diagonal weights are needed, so the deterministic rules of Sec.~\ref{sec:rules} remain applicable to coherent noise without modification, as long as the polarization direction is computed from the Pauli-diagonal part of the effective channel. For parameter estimation, by contrast, the residual intra-class coherences should be retained rather than twirled away: together with the diagonal weights they provide additional observables for coherent-noise fingerprints~\cite{ref:cohframe,ref:cohtomo}.

\subsection{Correlated noise}

For correlated noise the joint error distribution on \(n\) qubits need not factorize. Nevertheless, syndrome projection still assigns every multi-qubit error string to a logical Pauli class, and property~(i) of Lemma~\ref{lem:redistribution} depends only on the code and the fixed recovery, not on factorization of the noise; hence the block structure holds for arbitrary correlations. What changes is the computation of the weights. For independent noise they follow from Eqs.~\eqref{eq:rec-C1-X} and \eqref{eq:rec-C2-Z}; for correlated noise the probability of a logical class must be recomputed from the joint distribution,
\begin{equation}\label{eq:PrP}
\operatorname{Pr}(\bar P)=\sum_{E:\ \mathrm{class}(E)=\bar P}\operatorname{Pr}(E),
\end{equation}
where the sum runs over all physical error strings assigned to the logical class \(\bar P\). The code-induced redistribution is therefore still well defined, but the simple recursions no longer apply, and the quantitative crossover of Sec.~\ref{sec:threshold}, which assumes independent noise, is not directly valid. A full treatment requires the joint class probabilities and, for coherent correlations, the corresponding amplitudes. We regard this as the most important open direction, and we do not claim any quantitative result for correlated noise here.

\subsection{Inverse use: coherent-noise tomography}
\label{sec:inverse}

The persistence of the class assignment under correlations also makes redistribution useful as a probe: different codes and recovery conventions project the same physical noise onto different logical observables, so the diagonal weights of several code-derived channels constrain the underlying model. For coherent noise the diagonal weights carry syndrome-class probabilities while the residual intra-class coherences retain amplitude and phase information, which is why a tomographic protocol should preserve them rather than twirl them away~\cite{ref:cohframe,ref:cohtomo}.

A companion paper develops this inverse direction in full~\cite{ref:cohtomo}: it shows local identifiability of an eleven-parameter coherent-noise family from the fingerprint components of seven single-qubit processes, classifies the residual reflection-generated ambiguity as structural at any chain length, closes it with one two-qubit marginal channel, and propagates the recovered generator into decoder priors. We summarize that result here only to delimit the scope of the present work; no number quoted in this paper depends on it, and the forward results of Secs.~\ref{sec:redistribution}--\ref{sec:threshold} stand independently of Ref.~\cite{ref:cohtomo}.

\subsection{Summary of the boundary}

For a fixed code and recovery map, noise redistribution follows from syndrome projection and applies to arbitrary single-qubit noise and, at the level of class assignment, to correlated noise. The limits of the present quantitative analysis are equally clear: the closed-form recursions and the numerical crossover are derived for the three-qubit codes with fixed standard recoveries, under independent Pauli noise, with ideal encoding and recovery. Logically distinct recovery choices define different redistribution maps and must be analyzed separately. The conceptual message remains robust: stabilizer codes reshape noise through code structure and recovery, without requiring an explicit physical twirl.

\section{Conclusions and outlook}
\label{sec:conclusion}

We have formulated stabilizer-code QEC as a code- and recovery-dependent redistribution of noise. For a fixed recovery map, syndrome projection produces a logical channel that is block diagonal with respect to the syndrome partition, whose diagonal Pauli weights are fixed without any physical twirl, and concatenation becomes the iteration of that transformation. The three-qubit codes provide a closed-form realization under fixed standard recoveries. For amplitude damping the rules select \(C_2(\hat\sigma_y)C_1(\hat\sigma_x)\) within the family: one level converts the damping into a dominant logical-\(X\) channel at the price \(f_1=1-3\gamma/4\), and the second suppresses it to \(f_2=1-\tfrac94\gamma^2\), so that the concatenated protocol overtakes the unencoded channel below \(\gamma_c=0.2905\) (\(f_c=0.8485\))---provided the gates are accurate to a few times \(10^{-4}\), which is the honest measure of the ideal-recovery qualification. The same construction on longer repetition blocks raises the suppression order as \(\gamma^{w+1}\) while the gate budget rises as \(n^{L}\), so that the three-qubit block is the efficient choice per gate over almost the whole range where encoding pays; at a fixed total budget the five-qubit two-level cascade is ahead under the primitive counting of Eq.~\eqref{eq:grec}, and only a recovery overhead growing faster with \(w\) reverses that ordering.

Conceptually, the framework separates two design choices: selecting a code and selecting a recovery. The minimal examples remove the second freedom to make the mechanism explicit; when logically inequivalent recoveries are available, recovery selection can instead shape the effective channel for the next decoder. This is a design opportunity, not a claim that all larger codes have nonunique recoveries or identical closed-form behavior. Likewise, the three-qubit example is not presented as an inner-code prescription for surface-code architectures, where connectivity, circuit noise, and error locality must be analyzed explicitly, and where noise-adapted codes~\cite{ref:ad3qubit,ref:nsa,ref:adhw} may be preferable at the operating damping strength.

Three directions follow. First, the recovery-dependent class maps should be extended to larger codes, where the freedom of Remark~\ref{rem:recovery-design} becomes a genuine optimization variable rather than a fixed convention. Second, independent probabilities should be replaced by joint statistics, which requires the class sums of Eq.~\eqref{eq:PrP} and turns the boundary section into a quantitative statement. Third, noisy fault-tolerant circuits should be analyzed, since the sensitivity analysis of Sec.~\ref{sec:gateprec} shows that the gap between the ideal-recovery crossover and a device-relevant one is set entirely by gate accuracy~\cite{ref:noisystab,ref:corrsm,ref:ftconcat}. In the inverse direction, the companion work~\cite{ref:cohtomo} closes the identifiability question at the level of the noise model; together, the two directions connect the forward problem of designing a channel transformation with the reverse problem of inferring noise from the transformations produced by several codes.

\appendix

\section{Kraus-level derivation of Lemma~\ref{lem:redistribution}}
\label{app:kraus}

Let \(\mathcal{N}^{\otimes n}(\rho)=\sum_{\mu\nu}\chi_{\mu\nu}E_\mu\rho E_\nu^\dagger\) with Pauli operators \(E_\mu\) and syndrome \(s(\mu)\). Inserting into Eq.~\eqref{eq:eff} and using \(\Pi_sE_\mu P=\delta_{s,s(\mu)}\Pi_sE_\mu P\),
\begin{equation}
\mathcal{E}_{\mathrm{eff}}(\rho_L)=\sum_s\sum_{\substack{\mu\nu\\ s(\mu)=s(\nu)=s}}\chi_{\mu\nu}\,U_sE_\mu P\,\rho_L\,PE_\nu^\dagger U_s^\dagger .
\end{equation}
On the codespace, \(U_sE_\mu P=\bar P_\mu P\) up to a stabilizer element, which acts as the identity on \(P\); the phases are global and drop out of the channel. Hence
\begin{equation}
\begin{aligned}
\mathcal{E}_{\mathrm{eff}}(\rho_L)&=\sum_{\bar P\bar P'}\chi^{L}_{\bar P\bar P'}\,\bar P\rho_L\bar P'^{\dagger},\\
\chi^{L}_{\bar P\bar P'}&=\sum_{s}\!\!\sum_{\substack{\mu\in\mathcal{C}_{\bar P},\ \nu\in\mathcal{C}_{\bar P'}\\ s(\mu)=s(\nu)=s}}\!\!\chi_{\mu\nu},
\end{aligned}
\end{equation}
which is Eq.~\eqref{eq:chi} with the equal-syndrome constraint displayed. Three consequences follow at once.

(i) \emph{Block structure.} Every retained term satisfies \(s(\mu)=s(\nu)\). Since \(s\) is a function of the error, the double sum over \(\bar P,\bar P'\) never pairs different sectors, which is the block-diagonal statement of the lemma.

(ii) \emph{Diagonal weights.} Setting \(\bar P'=\bar P\) gives \(\chi^{L}_{\bar P\bar P}=\sum_{\mu\in\mathcal{C}_{\bar P}}\sum_{\nu\in\mathcal{C}_{\bar P}}\chi_{\mu\nu}\delta_{s(\mu)s(\nu)}\). For Pauli-diagonal noise (\(\chi_{\mu\nu}=p_\mu\delta_{\mu\nu}\)) the constraint is automatic and \(\chi^{L}_{\bar P\bar P}=\sum_{\mu\in\mathcal{C}_{\bar P}}p_\mu\), which is Eq.~\eqref{eq:classcount}. Note that for a general channel the class weight includes the intra-class coherences and not only the diagonal \(\chi_{\mu\mu}\).

(iii) \emph{Twirl independence.} The twirl \(\Lambda_L\mapsto\frac14\sum_{\bar Q}\bar Q\Lambda_L(\bar Q\rho\bar Q)\bar Q\) acts on the coefficient matrix as \(\chi^{L}_{\bar P\bar P'}\mapsto\chi^{L}_{\bar P\bar P'}\delta_{\bar P\bar P'}\), because conjugation by logical Paulis permutes the logical labels and averages the off-diagonal entries to zero, while the diagonal entries are invariant. Hence every \(w_{\bar P}\) is untouched. This is the precise content of~(iv).

For a non-Pauli physical channel the same derivation applies with \(\chi_{\mu\nu}\) the full process matrix; no assumption on the diagonality of \(\chi\) is used anywhere except in the last step of~(ii).

\section{Exact expansions for the amplitude-damping cascade}
\label{app:expansions}

With \(\gamma=1-s^2\), \(s=\sqrt{1-\gamma}\), every quantity below is an exact polynomial in \(s\) with rational coefficients; we quote the expansions in \(\gamma\).

\emph{Level one, \(C_2\) correcting \(\sigma_y\)} [Eqs.~\eqref{eq:lvl1a}--\eqref{eq:lvl1b}, through \(O(\gamma^4)\)]:
\begin{align*}
f_1&=1-\tfrac34\gamma+\tfrac{7}{64}\gamma^3-\tfrac{3}{256}\gamma^4,\\
p_{1x}&=\tfrac34\gamma-\tfrac{3}{16}\gamma^2-\tfrac{11}{64}\gamma^3-\tfrac{3}{128}\gamma^4,\\
p_{1y}&=\tfrac{3}{16}\gamma^2-\tfrac{5}{64}\gamma^3+\tfrac{3}{128}\gamma^4,\quad
p_{1z}=\tfrac{9}{64}\gamma^3+\tfrac{3}{256}\gamma^4 .
\end{align*}

\emph{Level two, \(C_1\) correcting \(\sigma_x\)} [Eqs.~\eqref{eq:lvl2w}--\eqref{eq:f2}, through \(O(\gamma^4)\)]:
\begin{align*}
f_2&=1-\tfrac94\gamma^2+\tfrac32\gamma^3+\tfrac{45}{128}\gamma^4,\\
p_{2x}&=\tfrac{27}{16}\gamma^2-\tfrac{27}{16}\gamma^3-\tfrac{63}{256}\gamma^4,\\
p_{2y}&=\tfrac{27}{32}\gamma^3-\tfrac{225}{256}\gamma^4,\quad
p_{2z}=\tfrac{9}{16}\gamma^2-\tfrac{21}{32}\gamma^3+\tfrac{99}{128}\gamma^4 .
\end{align*}

\emph{Crossover.} The difference of Eqs.~\eqref{eq:f2} and \eqref{eq:f0exact} is
\begin{equation*}
f_2-f_0=\tfrac12\gamma-\tfrac{35}{16}\gamma^2+\tfrac{49}{32}\gamma^3+\tfrac{95}{256}\gamma^4,
\end{equation*}
whose leading-order root is \(8/35\). Solving the exact polynomial in \(s\) gives \(\gamma_c=0.29049987\), \(f_c=0.84853418\).

\emph{Alternative cascade} \(C_1(\hat\sigma_x)\to C_2(\hat\sigma_z)\) [Rule~\ref{rule:initialY}]:
\begin{equation*}
1-f_2^{\mathrm{alt}}=\tfrac{63}{16}\gamma^2-\tfrac{33}{8}\gamma^3+O(\gamma^4),
\end{equation*}
i.e.\ the same second-order suppression with a coefficient \(1.75\) times larger than the \(9/4\) of Eq.~\eqref{eq:f2}.

\emph{Substitution with fixed output labels} (Sec.~\ref{sec:ycorr}): applying \(p_z\leftrightarrow p_y\) to the right-hand side of Eq.~\eqref{eq:rec-C2-Z} without exchanging \(p_{1y}\) and \(p_{1z}\) yields a level-one channel whose \(y\) and \(z\) weights are exchanged with respect to Eq.~\eqref{eq:lvl1b}; continuing the cascade with that channel gives \(\gamma_c=0.29205\) instead of \(0.29050\).

\end{document}